\documentclass[12pt]{amsart}
\usepackage{amsmath}
\usepackage{latexsym}
\usepackage{amsfonts}
\usepackage{amssymb}
\usepackage{bbm,dsfont}
\usepackage{color}

\usepackage{ale}

\newtheorem{proposition}{Proposition}
\newtheorem{theorem}{Theorem}
\newtheorem{lemma}{Lemma}
\newtheorem{corollary}{Corollary}

\theoremstyle{definition}

\newtheorem{remark}{Remark}
\newtheorem{example}{Example}
\newtheorem{definition}{Definition}

\newcommand{\eps}{\varepsilon}

\newcommand{\R}{\mathbb R} 
\newcommand{\C}{\mathbb C} 
\newcommand{\N}{\mathbb N} 

\newcommand{\h}{\mathcal{H}}
\newcommand{\hi}{\mathcal{H}} 
\newcommand{\lh}{\mathcal{L(H)}} 
\newcommand{\sh}{\mathcal{S(H)}} 

\newcommand{\no}[1]{\left\|#1\right\|} 
\newcommand{\tr}[1]{\mathrm{tr}\left[#1\right]} 
\newcommand{\LUB}[1]{{\rm LUB}\, \left\{ #1 \right\}} 

\newcommand{\nil}{O} 

\newcommand{\Eo}{\mathsf{E}} 
\newcommand{\Fo}{\mathsf{F}} 
\newcommand{\Go}{\mathsf{G}} 
\newcommand{\Po}{\mathsf{P}} 
\newcommand{\PP}{\mathsf{P}} 

\newcommand{\bor}[1]{\mathcal{B}(#1)} 
\newcommand{\linf}[1]{L^\infty \left( #1 \right)}

\newcommand{\supp}{{\rm supp}\,}

\newcommand{\de}{\, d}

\newcommand{\frecc}{\rightarrow}

\newcommand{\id}{{\rm id}} 

\begin{document}

\title[]{Commutative POVMs and Fuzzy Observables}

\author[Ali]{S. Twareque Ali}
\address{S. Twareque Ali, Department of Mathematics and Statistics,
Concordia University, Montreal, Quebec, Canada H3G 1M8}
\email{stali@mathstat.concordia.ca}

\author[Carmeli]{Claudio Carmeli}
\address{Claudio Carmeli, Dipartimento di Fisica, Universit\`a di Genova
and I.N.F.N., Sezione di Genova, Via Dodecaneso 33, 16146 Genova, Italy}
\email{carmeli@ge.infn.it}

\author[Heinosaari]{Teiko Heinosaari}
\address{Teiko Heinosaari, Niels Bohr Institute, Copenhagen University,
Denmark and Department of Physics, University of Turku, Finland}
\email{heinosaari@gmail.com}

\author[Toigo]{Alessandro Toigo}
\address{Alessandro Toigo, Dipartimento di Informatica, Universit\`a di Genova,
Via Dodecaneso 35  and I.N.F.N., Sezione di Genova, Via Dodecaneso 33,
16146 Genova, Italy}
\email{toigo@ge.infn.it}

\begin{abstract}
In this paper we review some properties of fuzzy observables, mainly as
realized by commutative positive operator valued measures. In this context
we discuss two representation theorems for commutative positive operator
valued measures in terms of projection valued measures and describe, in some
detail, the general notion of fuzzification. We also make some related observations
on joint measurements.
\end{abstract}

\maketitle

\section{Introduction}

Fuzzy observables arise in several ways in quantum formalism.
Their history goes back to the 1970s, and they were one of the first concrete examples
of genuine positive operator valued measures (POVMs) in quantum mechanics.
For commutative POVMs, this concept gives useful representation theorems in terms
of projection valued measures (PVMs). On the other hand, fuzzification is a natural
way to model reduction in the statistical description of the system.  Fuzzy observables
are also intimately related to the topic of joint measurements.

In this article we discuss all these aspects of fuzzy observables. In Section \ref{sec:history}
we give a brief historical overview of fuzzy observables. Section \ref{sec:povm} contains
some basic definitions and fixes the notation. In Section \ref{sec:commutative} we present
two representation theorems for commutative POVMs. In Section \ref{sec:fuzzy} we
explain the connection of these results to the concept of fuzzy observables.  Finally, in
Section \ref{sec:joint} we present some remarks related to joint measurements.
\blfootnote{\textsc{This article is dedicated to Pekka Lahti on the occasion of his 60th birthday.}}

\section{A quick historical overview}\label{sec:history}

The appearance of positive operator valued measures in studying of the foundations
of quantum mechanics followed along several lines. The axiomatization of quantum
mechanics, built upon an analysis of the measurement process,  as undertaken by
G\"unther Ludwig and the Marburg school in the 1960's and 70's
 \cite{Ludwig64,Ludwig70}, led to the notion of generalized observables, which
 went beyond the orthodox concept of an observable in quantum mechanics. An observable
in this more general setting turned out to be a positive operator valued measure on a
Hilbert space. The $\sigma$-algebra of this measure was generally built on the set of
possible values of the observable or observables, the implication being that it was
theoretically possible to discuss joint measurements of incompatible observables.
Thus, in orthodox quantum mechanics the observables of position and momentum are
each characterized by projection valued measures, i.e., the spectral projectors of the
corresponding self adjoint operators, defined on the space of possible observed
values of position and momentum,
and they are not simultaneously measurable. In the more general setting, it is entirely
possible to discuss the joint measurement of these two observables, using a positive
operator valued measure defined on the joint set of values of classical position and
momentum. Of course, in such a joint measurement neither position nor momentum can
assume sharp values, implying  necessarily a certain {\em fuzzification\/.} This line of
thought leads to a one type of formulation of Heisenberg's uncertainty principle -  a
condition ensuring that mutually exclusive measurements of position and momentum
can be reconciled if an appropriate trade-off is accepted. For a recent exposition of
this topic, we refer to \cite{BuHeLa07}.

The second line of development, which led to the use of POVMs was the generalization
of the notion of a conditional expectation in probability theory. An algebraic formulation of the
notion of conditional expectation, in the context of classical, i.e., commutative, probability
theory was worked out in the 1950's \cite{Moy54, NaTu54}, which was further developed
and extended to non-commutative probability theories, in which the classical
algebra of stochastic variables was replaced by a non-abelian von
Neumann algebra. Such a generalization, apart from its intrinsic mathematical
interest,  was then shown to be  useful in discussions of the measurement process in
quantum mechanics \cite{NaUm62,Umegaki54}.  Later,  Davies and Lewis, in
developing an operational basis for quantum measurement theory and a related
theory of quantum stochastic processes,  \cite{DaLe70,Davies70, QTOS76} gave
a definition of a conditional expectation, as the dual to the concept of a measurement.
This had the virtue of extending the von Neumann collapse postulate to observables
with continuous spectra. Each measurement process then singled out an observable,
defined as a POVM.

There is at least one more significant and independent 
line of development involving
POVMs,  which  started in the early 1970s, and most notably expounded and 
developed in the works of 
Holevo and Helstrom  \cite{He76,Ho82}.  Here the notion of a POVM appeared in the
then emerging field of quantum information and �statistical decision theory,
essentially independently of the ideas of Ludwig on quantum foundations and
Davies and Lewis on repeated measurements. This pioneering work led to the
recognition of the fact that the statistics of  optimal quantum
measurements is often obtained by using what were called {\em non-orthogonal POVMs}.
Simultaneously, the  relevance of  Naimark's extension theorem (from POVMs to projection 
valued measures),  was first pointed out in this context.

The idea that POVMs could in some sense describe unsharp measurements was already
present in the work of Davies and Lewis. These ideas were further developed to a
more complete theory of fuzzy observables in \cite{AlEm74, SQMQS84}. In the
succeeding years the connection between unsharp or fuzzy observables and POVMs
has been worked out in great detail by a number of authors. For a systematic
presentation of various aspects of POVMs, we refer to \cite{OQP97}.

\section{Positive operator valued measures}\label{sec:povm}

Quantum observables are mathematically described by positive operator valued measures.
In this section we shortly recall this concept and fix the notation.

Let $\Omega$ be a Hausdorff locally compact second countable topological
space (lcsc space, for short). We use the following notations:
\begin{itemize}
\item $\bor{\Omega}$ is the Borel $\sigma$-algebra of $\Omega$
\item $C_0 (\Omega)$ is the Banach space of countinuous complex functions on
$\Omega$ vanishing at infinity, endowed with the
uniform norm
\item $M(\Omega)$ is the Banach space of complex measures on $\Omega$ endowed with the total variation norm
\item $M(\Omega)^+$ is the subset of positive elements in $M(\Omega)$
\item $P(\Omega)\subset M(\Omega)^+$ is the subset of probability measures
\end{itemize}

We will regard an element $\mu\in M(\Omega)$ both as a $\sigma$-additive mapping
$\mu : \bor{\Omega} \frecc \C$ and as a bounded linear functional
$\mu : C_0 (\Omega) \frecc \C$. In the latter case, we have the isometric isomorphism $M(\Omega) = C_0 (\Omega)^\ast$. In general, if $\linf{\Omega}$ denotes the Banach space of bounded measurable complex functions on $\Omega$ endowed with the uniform norm, an element $\mu\in M(\Omega)$ determines a bounded linear
functional $\mu : \linf{\Omega} \frecc \C$ by means of the formula
\begin{equation*}
\mu (\phi) = \int_\Omega \phi ( x) \de \mu ( x) \quad \forall \phi \in \linf{\Omega} .
\end{equation*}
We let $B(\Omega)$ be the unit ball in $M(\Omega)$, and $B(\Omega)$ is endowed with the weak*-topology. Since $C_0 (\Omega)$ is separable, $B(\Omega)$ is a compact metrizable space. Moreover, $B(\Omega)^+ := M(\Omega)^+ \cap B(\Omega)$ is closed in $B(\Omega)$, and  $P(\Omega)$ is closed in $B(\Omega)^+$ if and only if $\Omega$ is compact.

Let $\hi$ be a complex separable Hilbert space. We denote by $\lh$ the Banach space
of bounded operators on $\hi$ with the uniform norm.  The convex set of positive trace one elements in $\lh$ is denoted by $\sh$.

A {\em positive operator valued measure} (POVM) $\Eo$ on $\Omega$ with values
in $\hi$ can be defined in two equivalent ways:
\begin{itemize}
\item[\rm (a)] as a mapping $\Eo : \bor{\Omega} \frecc \lh$ such that
\begin{itemize}
\item[\rm (i)] $\Eo(X) \geq \nil$ for all $X\in\bor{\Omega}$
\item[\rm (ii)] $\Eo(\Omega) = I$
\item[\rm (iii)] for all disjoint finite or denumerable sequences $\{ X_i \}_{i\in I}$ of
sets in $\bor{\Omega}$,
\begin{equation*}
\Eo \left( \cup_{i\in I} X_i \right) = \sum\nolimits_{i\in I} \Eo(X_i) ,
\end{equation*}
where the sum converges in the weak (or, equivalently, ultraweak or strong) operator topology;
\end{itemize}
\item[\rm (b)] as a linear mapping $\Eo : C_0 (\Omega) \frecc \lh$ such that
\begin{itemize}
\item[\rm (i)] $\Eo (f) \geq \nil$ if $f\geq 0$
\item[\rm (ii)] If $\{ f_n \}_{n\in \N}$ is a sequence of positive functions in
$C_0 (\Omega)$ such that $f_n \uparrow 1$, then $\Eo (f_n) \uparrow I$.
\end{itemize}
\end{itemize}

The link between the two definitions is established by the following integral formula
\begin{equation}\label{link}
\Eo (f) = \int_\Omega f( x) \de \Eo ( x) \, ,
\end{equation}
true for all $f\in C_0 (\Omega)$. Here in the left hand side $\Eo$ is defined according to {\rm (b)}, and in the right
hand side definition  {\rm (a)} is used. The integral in eq.~(\ref{link}) has to be understood in the following way. If $\Eo$ is a POVM in the sense of definition {\rm (a)}, then for all
$T\in\sh$ we can define a probability measure $p^{\Eo}_T \in P(\Omega)$, given by
\begin{equation}\label{eq:prob}
p^{\Eo}_T (X) = \tr{T\Eo (X)} \quad \forall X\in\bor{\Omega} \, .
\end{equation}
Eq.~(\ref{link}) then reads
\begin{equation*}
\tr{T \Eo (f)} = \int_\Omega f( x) \de p^{\Eo}_T ( x) \quad \forall T\in\sh .
\end{equation*}
Note that if $\phi\in\linf{\Omega}$, we can define in the same way the bounded operator
\begin{equation}\label{link2}
\Eo (\phi) = \int_\Omega \phi ( x) \de \Eo ( x) \, .
\end{equation}
We see from eq.~(\ref{link2}) that $\no{\Eo(\phi)} \leq \no{\phi}$ for every $\phi \in \linf{\Omega}$. In particular, $\Eo : \linf{\Omega} \frecc \lh$ is continuous. We refer to \cite{NST66} for a more detailed discussion of these facts.

In quantum formalism, the set $\sh$ represents the \emph{state space} of a quantum system. Quantum observables are, on the other hand, described by positive operator valued measures. The basic numerical predictions of quantum mechanics - the measurement outcome probability distributions - are given by the trace formula \eqref{eq:prob}. Namely, the number $p^{\Eo}_T (X)$ is interpreted as the probability that a measurement outcome $ x\in X$ occurs when a measurement of the observable $\Eo$ is performed and the system is in the state $T$.

\begin{remark}
If $\Omega$ is compact, then $C_0 (\Omega) = C (\Omega)$, the space of
continuous functions on $\Omega$. Item {\rm (ii)} in definition {\rm (b)} is
equivalent to $\Eo (1_{\Omega}) = I$, where $1_{\Omega}$ is the constant function $1_{\Omega}( x)\equiv 1$.
\end{remark}

For a subset $\mathcal{M}\subset\lh$, we denote by $\mathcal{M}^\prime$ the commutant of $\mathcal{M}$ in $\lh$. If $\Eo$ is a POVM, we denote
\begin{equation}\label{def-E'}
\Eo^\prime = \{ \Eo (f) \mid f\in C_0 (\Omega) \}^\prime .
\end{equation}

\begin{proposition}
Let $\Eo$ be a POVM. Then $\Eo^\prime = \{ \Eo (X) \mid X\in \bor{\Omega} \}^\prime$
\end{proposition}

\begin{proof}
If $U$ is unitary, then $U\in\Eo^\prime$ if and only if the POVMs $f\mapsto \Eo (f)$ and $f\mapsto \Eo_U (f) = U \Eo (f) U^\ast$ are equal. This implies that $\Eo (X) = \Eo_U (X) = U \Eo (X) U^\ast$ for all $X\in\bor{\Omega}$. Therefore, $U\in\Eo^\prime$ if and only if $U\in \{ \Eo (X) \mid X\in \bor{\Omega} \}^\prime$. The claim then follows by decomposing each $A\in\lh$ as a linear combination of four unitaries $\{ U_1 \ldots U_4 \}$ such that $A^\prime = \{ U_1 \ldots U_4 \}^\prime$.
\end{proof}

\section{Commutative observables}\label{sec:commutative}

To emphasize the underlying physical context, we will use the term observable instead of POVM. This section, however, is concerned with certain mathematical representations of observables. The physical interpretation of these results are then discussed in Section \ref{sec:fuzzy}.

\begin{definition}
An observable $\Eo$ is {\em commutative} if $\Eo (X) \in \Eo^\prime$ for all $X\in\bor{\Omega}$, or, equivalently, if $\Eo (f) \in \Eo^\prime$ for all $f\in C_0 (\Omega)$.
\end{definition}

A special case of a commutative observable is a {\em projection valued measure} (PVM). It is a mapping $\Po : \bor{\Omega}\frecc \lh$ which in definition {\rm (a)} satisfies the
additional property
\begin{equation*}
\Po(X)\Po(Y) = \Po (X\cap Y)
\end{equation*}
for all $X,Y\in\bor{\Omega}$. It is clear that commutativity of $\Po$ follows from this
condition. When using definition (b) this additional condition reads
\begin{equation*}
\Po (f) \Po (g) = \Po (fg)
\end{equation*}
for all $f,g \in C_0 (\Omega)$. Projection valued measures are commonly referred 
as \emph{sharp observables}.

As we show in this section, each commutative observable has at least two 
possible representations in terms of  sharp observables. This topic has been recently 
studied by Jen\v{c}ov\'{a} 
and Pulmannov\'{a} in \cite{JePu07},\cite{JePu09}, where related results have been 
obtained. Our first representation theorem for commutative observables (Theorem \ref{th:second}) 
goes back to Holevo \cite{Ho72}, who proved it in the general case in which $\Omega$ is 
an arbitrary measurable space. It claims that, if $\Eo$ is a commutative observable, then 
there is a {\em unique} sharp observable $\PP$ on the set of measures $B(\Omega)^+$ 
such that a kind of {\em canonical} spectral decomposition of $\Eo$ 
holds (eqs.~\eqref{eq:rep-1} and \eqref{eq:rep-2}). Actually, $\PP$ vanishes 
outside the set of probability measures $P(\Omega)$. As in the usual spectral theorem 
for normal operators, $\Eo$ and $\PP$ have the same commutant in $\lh$. Under our 
assumption that $\Omega$ is a lcsc space, we will work out a different and more 
easy proof of Holevo's result, which simply makes use of elementary facts from 
$C^\ast$-algebra theory. A second representation 
theorem (Theorem \ref{th:third}) which, as explained below, is  in a way  complementary to the first, was 
proved in \cite{Ali82}. We present it here  without proof, describing only its content.

\begin{theorem}\label{th:first}
Let $\Eo : C_0 (\Omega) \frecc \lh$ be a commutative observable. Then there exists a triple
$(\Sigma,\PP,\mu)$ such that
\begin{itemize}
\item[\rm (i)] $\Sigma$ is a compact metrizable space;
\item[\rm (ii)] $\PP$ is a sharp observable (i.e. PVM) on $\Sigma$ with values in $\hi$ and $\supp \PP = \Sigma$;
\item[\rm (iii)] $\mu : \Sigma \frecc B(\Omega)^+$ is a weak*-continuous injective mapping;
\item[\rm (iv)] for all $f\in C_0 (\Omega)$
\begin{equation}\label{eq. princ.}
\Eo(f) = \int_\Sigma [\mu (h)] (f) \de \PP (h).
\end{equation}
\end{itemize}

If $( \tilde{\Sigma},\tilde{\PP},\tilde{\mu} )$ is another triple satisfying
items {\rm (i)-(iv)} above, then there exists a homeomorphism
$\Phi : \Sigma \frecc \tilde{\Sigma}$ such that $\tilde{\PP} = \PP \circ \Phi^{-1}$ and
$\mu = \tilde{\mu} \circ \Phi$.

Moreover, we have the equality $\PP^\prime = \Eo^\prime$.
\end{theorem}

\begin{proof}
We denote by $C^\ast (\Eo)$ the commutative $C^\ast$-algebra generated in $\lh$ by the set $\{ \Eo (f) \mid f\in C_0 (\Omega) \}$ and $I$. Since $C_0 (\Omega)$ is separable and the mapping $f\mapsto \Eo (f)$ is continuous, the Banach space $C^\ast (\Eo)$ is separable.

Let $\Sigma$ be the maximal ideal space of $C^\ast (\Eo)$. Then $\Sigma$ is a compact metrizable space. We denote the Gelfand transform of $C^\ast (\Eo)$ by $\Gamma$, i.e.,
\begin{equation*}
\Gamma : C^\ast (\Eo) \frecc C(\Sigma), \qquad \Gamma A (h) =
h(A) \qquad \forall A \in C^\ast (\Eo), \, h\in\Sigma \, .
\end{equation*}
Thus, $\Gamma$ is a $C^\ast$-algebra isomorphism
between $C^\ast (\Eo)$ and $C(\Sigma)$.

For all $h\in\Sigma$, we define a mapping $\mu(h)$ from $C_0 (\Omega)$  to $\C$ by
\begin{equation*}
[\mu(h)] (f) := [\Gamma \Eo (f)] (h) \quad \forall f\in C_0 (\Omega) \, .
\end{equation*}
From the properties of $\Gamma$ and $\Eo$ follows that $\mu(h)$ is a positive bounded linear functional on $C_0 (\Omega)$,
hence $\mu (h) \in M(\Omega)^+$. Moreover,
\begin{equation*}
\left| [\mu(h)] (f) \right| \leq \no{\Gamma}
\no{\Eo (f)} \leq \no{f} \, ,
\end{equation*}
which implies that $\no{\mu(h)} \leq 1$.

We have thus defined a mapping $\mu$ from $\Sigma$ to $M(\Omega)$.
The weak*-continuity of $\mu$ is clear from the definition. Since the subset
$\{ [\mu(\cdot)] (f) \mid f\in C_0 (\Omega) \}$ together with the constant function $1_{\Omega}$ generates the algebra $C(\Sigma)$, it separates the points of $\Sigma$. Therefore, $\mu : \Sigma \frecc M(\Omega)$ is injective.

Since $\Gamma^{-1}$ is a $\ast$-representation of $C(\Sigma)$ in $\lh$, there exists a sharp observable (PVM) $\PP : \bor{\Sigma} \frecc \lh$ such that
\begin{equation*}
\Gamma^{-1} (\phi) = \int_\Sigma \phi (h) \de \PP (h) \quad \forall \phi \in C(\Sigma)
\end{equation*}
(see e.g. Theorem 2.5.5 in \cite{CAOT89}).
Since $\Gamma^{-1}$ is injective, $\supp \PP = \Sigma$.
Moreover,
\begin{equation*}
\Eo(f) = \Gamma^{-1} \Gamma \Eo (f) = \int_\Sigma [\mu (h)] (f) \de \PP (h) \, ,
\end{equation*}
and item {\rm (iv)} follows.

Now suppose $( \tilde{\Sigma},\tilde{\PP},\tilde{\mu} )$ is another triple
satisfying items {\rm (i)-(iv)} in the theorem. Let
\begin{equation*}
\Pi (\tilde{\phi}) = \int_{\tilde{\Sigma}} \tilde{\phi} (\tilde{h}) \de \tilde{\PP} (\tilde{h})
\quad \forall \tilde{\phi} \in C(\tilde{\Sigma}) \, .
\end{equation*}
Then $\Pi$ is a $\ast$-homomorphism of $C(\tilde{\Sigma})$ in $\lh$, which is injective
since $\supp \tilde{\PP} = \tilde{\Sigma}$. Moreover,
\begin{equation*}
\Pi \left( [\tilde{\mu} (\cdot)] (f) \right) = \Eo(f) \quad \forall f\in C_0 (\Omega) .
\end{equation*}
By injectivity of $\tilde{\mu}$, the set of functions $\{ [\tilde{\mu} (\cdot)] (f) \mid f\in C_0 (\Omega) \}$ separates the points of $\tilde{\Sigma}$, and it is closed under complex conjugation since
$\tilde{\mu} (\tilde{h}) \in M(\Omega)^+$ for all $\tilde{h} \in \tilde{\Sigma}$. The functions $\{ [\tilde{\mu} (\cdot)] (f) \mid f\in C_0 (\Omega) \}$ and $1_{\Omega}$ thus generate the algebra $C(\tilde{\Sigma})$ by Stone-Weierstrass theorem. It follows that $\Pi ( C(\tilde{\Sigma}) ) = C^\ast (E)$.

The composition $\Gamma \Pi$ is thus a $\ast$-isomorphism of $C(\tilde{\Sigma})$
into $C(\Sigma)$, hence there exists a homeomorphism $\Phi : \Sigma \frecc \tilde{\Sigma}$
such that
\begin{equation*}
\Gamma \Pi (\tilde{\phi}) = \tilde{\phi} \circ \Phi \quad \forall \tilde{\phi} \in C(\tilde{\Sigma})
\end{equation*}
(see Theorem 2.1 in \S VI of \cite{CFA90}). We thus have, for all $f\in C_0 (\Omega)$,
\begin{equation*}
[\mu(\cdot)] (f) = \Gamma \Pi \left( [\tilde{\mu} (\cdot)] (f) \right) =
[\tilde{\mu} \circ \Phi (\cdot)] (f) \, .
\end{equation*}
Moreover,
\begin{equation*}
\int_{\tilde{\Sigma}} \tilde{\phi} (\tilde{h}) \de \tilde{\PP} (\tilde{h}) = \Pi (\tilde{\phi}) = \Gamma^{-1} \left[ \Gamma \Pi (\tilde{\phi}) \right] = \int_{\Sigma} \tilde{\phi} (\Phi(h)) \de \PP (h) = \int_{\tilde{\Sigma}} \tilde{\phi} (\tilde{h}) \de (\PP \circ \Phi^{-1}) (\tilde{h})
\end{equation*}
for all $\tilde{\phi} \in C(\tilde{\Sigma})$, which implies $\tilde{\PP} = \PP \circ \Phi^{-1}$.

Finally, $S\in\Eo^\prime$ if and only if $S\in C^\ast (\Eo)^\prime$, and, by spectral theory, $C^\ast (\Eo)^\prime = \PP^\prime$, so that the last claim is proved.
\end{proof}

\begin{proposition}\label{prop:compact}
Suppose that $\Omega$ is compact. Then in the situation of Theorem \ref{th:first}, $\mu (h)$ is a probability measure for all $h\in\Sigma$.
\end{proposition}

\begin{proof}
For each $h\in\Sigma$, we get
\begin{equation*}
[\mu(h)](1_{\Omega}) = [\Gamma \Eo(1_{\Omega})](h) = \Gamma I(h)=h(I)=1 \, .
\end{equation*}
\end{proof}

Suppose $\Sigma$ is a lcsc space. We say that a mapping $\mu : \Sigma \frecc M(\Omega)^+$ is {\em weak*-measurable} if the mappings
\begin{equation*}
\Sigma \ni h \to [\mu (h)] (f) \in \C
\end{equation*}
are measurable for all $f\in C_0 (\Omega)$. The following auxiliary result is needed later.

\begin{proposition}\label{misMark}
Let $\Sigma$ be a lcsc space and $\mu$ a mapping from $\Sigma$ to $M(\Omega)^+$. Then the following are equivalent:
\begin{itemize}
\item[(a)] $\mu$ is  weak*-measurable;
\item[(b)] the mappings $\Sigma \ni h \to [\mu (h)] (X) \in \R$ are measurable for all $X\in \bor{\Omega}$.
\end{itemize}
If $\mu(\Sigma) \subset B(\Omega)^+$, then the above two conditions are also equivalent to
\begin{itemize}
\item[(c)] $\mu : \Sigma \frecc B(\Omega)^+$ is a measurable map.
\end{itemize}
\end{proposition}

\begin{proof}
Suppose (b) holds. Fix $f\in C_0 (\Omega)$. For $n \in \N$, let $\{ E_i \}_{i\in I}$ be a finite partition of $\Omega$ into disjoint Borel sets such that $|f (x) - f (y)| \leq 1/n$ for all $ x, y \in E_i$ and $i\in I$. Choose $x_i \in E_i$, and let $\phi_n = \sum_i f (x_i) 1_{E_i}$. Then the mapping $h \mapsto [\mu (h)] (\phi_n)$ is measurable, and
\begin{equation*}
\left| [\mu (h)] (f) - [\mu (h)] (\phi_n) \right| \leq \no{\mu(h)}/n \, .
\end{equation*}
Hence the mapping $h \mapsto [\mu (h)] (f)$ is measurable being the pointwise limit of measurable functions.

Suppose then that (a) holds, i.e., $\mu$ is weak*-measurable. Let
\begin{equation*}
L^\infty_\mu = \left\{ \phi\in\linf{\Omega} \mid h \mapsto [\mu (h)] (\phi) \textrm{ is a measurable mapping} \right\} .
\end{equation*}
If $\{ \phi_n \}_{n\in\N} \subset L^\infty_\mu$, $\phi\in\linf{\Omega}$ are such that $\phi_n \uparrow \phi$, then $[\mu (h)] (\phi_n) \uparrow [\mu (h)] (\phi)$ for all $h$ by dominated convergence theorem, hence $h \mapsto [\mu (h)] (\phi)$ is measurable. This shows that $L^\infty_\mu$ is a monotone class in $\linf{\Omega}$. Since $C_0 (\Omega) \subset L^\infty_\mu$, $L^\infty_\mu = \linf{\Omega}$ by Proposition 6.2.9 in \cite{AN89}. In particular, $h \mapsto [\mu (h)] (X) = [\mu (h)] (1_X)$ is measurable for all $X\in\bor{\Omega}$.

Now assume that $\mu(\Sigma) \subset B(\Omega)^+$. Since the mapping $B(\Omega)^+ \ni \nu \mapsto \nu (f) \in \C$ is continuous, it is clear that (c) implies (a). On the other hand, suppose that (a) holds. Since $B(\Omega)^+$ is second countable, there exist sequences $\{ \nu_n \}_{n\in\N}$ in $B(\Omega)^+$, $\{ \eps_n \}_{n\in\N}$ in $\R_+$, $\{ k_n \}_{n\in\N}$ in $\N$, and, for all $n\in\N$, $\{ f^n_i \}_{i=1}^{k_n}$ in $C_0 (\Omega)$, such that the open sets
\begin{equation*}
U_n = \{ \nu\in B(\Omega)^+ \mid |\nu (f^n_i) - \nu_n (f^n_i) | < \eps_n \textrm{ for all } i = 1,2 \ldots k_n \}
\end{equation*}
form a denumerable basis for the topology of $B(\Omega)^+$. By (a), $\mu^{-1} (U_n) \in \bor{\Sigma}$ for all $n$, hence $\mu^{-1} (U) \in \bor{\Sigma}$ for all open $U\subset B(\Omega)^+$. This implies that $\mu$ is measurable.
\end{proof}

\begin{corollary}\label{Prop. misurab. di P}
The set $P(\Omega)$ of probability measures is a measurable subset in $B(\Omega)^+$.
\end{corollary}

\begin{proof}
Take in the above proposition $\Sigma = B(\Omega)^+$ and $\mu$ the identity mapping, which is weak*-measurable being clearly weak*-continuous. The set $P(\Omega) = \{ \nu \in B(\Omega)^+ \mid \nu (\Omega) = 1 \}$ is thus measurable, since it is the inverse image of $1$ under the measurable mapping $\mu \mapsto \mu (\Omega)$.
\end{proof}

\begin{proposition}\label{Prop. supp. di P}
Let $\Eo$ and $(\Sigma, \PP, \mu)$ be as in Theorem \ref{th:first}. Then $\mu (h)$ is a probability measure for $\PP$-almost all $h\in\Sigma$.
\end{proposition}

\begin{proof}
If $\{ f_n \}_{n\in\N}$ is a sequence of positive functions in $C_0 (\Omega)$ such that $f_n \uparrow 1_{\Omega}$, we have
\begin{equation*}
\int_\Sigma [\mu (h)] (f_n) \de \PP (h) = \Eo(f_n) \uparrow I .
\end{equation*}
On the other hand, $[\mu (h)] (f_n) \uparrow [\mu (h)] (\Omega)$ for all $h\in\Sigma$, so that, by monotone convergence theorem and uniqueness of the limit,
\begin{equation}\label{normalisation}
\int_\Sigma [\mu (h)] (\Omega) \de \PP (h) = I .
\end{equation}
For $\varepsilon > 0$, let $X_\varepsilon = \{ h\in \Sigma \mid [\mu (h)] (\Omega) < 1 - \varepsilon \}$. Multiplying eq.~(\ref{normalisation}) by $\PP (X_\varepsilon)$, we obtain
\begin{equation*}
\PP (X_\varepsilon) = \int_{X_\varepsilon} [\mu (h)] (\Omega) \de \PP (h) \leq ( 1 - \varepsilon ) \PP (X_\varepsilon) ,
\end{equation*}
hence $\PP (X_\varepsilon) = 0$. The claim then follows.
\end{proof}

\begin{proposition}\label{caratt. di E sui Borel}
Let $\Eo$ and $(\Sigma, \PP, \mu)$ be as in Theorem \ref{th:first}. Then
\begin{equation}\label{eq:borel}
\Eo (X) = \int_\Sigma [\mu(h)] (X) \de \PP (h) \quad \forall X\in\bor{\Omega} \, .
\end{equation}
\end{proposition}

\begin{proof}
For every $X\in\bor{\Omega}$, define
\begin{equation*}
\Fo (X) = \int_\Sigma [\mu(h)] (X) \de \PP (h) \, .
\end{equation*}
Then $\Fo$ is an observable (in the sense of definition {\rm (a)}). If $C\subset \Omega$ is compact, and $\{ f_n \}_{n\in\N}$ is a sequence in $C_0 (\Omega)$ such that $f_n \downarrow 1_C$ (such a sequence exists by Lemma 6.2.8 in \cite{AN89}), then by dominated convergence theorem $\Eo (f_n) \downarrow \Eo (C)$, $[\mu(h)] (f_n) \downarrow [\mu(h)] (C)$ for all $h$, and $\int_\Sigma [\mu(h)] (f_n) \de \PP (h) \downarrow \Fo (C)$. By eq.~(\ref{eq. princ.}) and uniqueness of the limit, $\Eo(C) = \Fo(C)$.

As $\Omega$ is a lcsc space, every probability measure $p^{\Eo}_T$, with $T\in\sh$, is regular (see e.g. Chapter 7 in \cite{RA99}). Therefore, the observable $\Eo$ has the following regularity property:
\begin{equation*}
\Eo(X) = \LUB{\Eo(C)\mid C \textrm{ compact subset of } X } \, .
\end{equation*}
It then follows that $\Eo(X) = \Fo(X)$ for all $X\in\bor{\Omega}$.
\end{proof}

If $\PP$ is a PVM on $B(\Omega)^+$ such that $\PP (B(\Omega)^+ \setminus P(\Omega)) = \nil$, then it is easy to check that the formula
\begin{equation*}
\Eo (f) = \int_{B(\Omega)^+} \mu(f) \de \PP(\mu) \quad \forall f\in C_0 (\Omega)
\end{equation*}
defines a commutative observable $\Eo$ on $\Omega$. Collecting the above results, we obtain the following converse of this fact.

\begin{theorem}\label{th:second}
Let $\Eo$ be a commutative observable on $\Omega$. Then there exists a unique sharp observable $\PP$ on $B(\Omega)^+$ such that
\begin{equation}\label{eq:rep-1}
\Eo (f) = \int_{B(\Omega)^+} \mu (f) \de \PP (\mu) \quad \forall f\in C_0 (\Omega) \, ,
\end{equation}
or, equivalently,
\begin{equation}\label{eq:rep-2}
\Eo (X) = \int_{B(\Omega)^+} \mu (X) \de \PP (\mu) \quad \forall X\in \bor{\Omega} \, .
\end{equation}
Moreover, $B(\Omega)^+ \setminus P(\Omega)$ is a $\PP$-null subset of $B(\Omega)^+$, and $\Eo^\prime = \PP^\prime$.
\end{theorem}

\begin{proof}
Apply Theorem \ref{th:first}, and observe that $\mu : \Sigma \frecc B(\Omega)^+$, being a continuous injective mapping between compact metric spaces, is a homeomorphism of $\Sigma$ onto $\mu(\Sigma)$. Eq.~(\ref{eq:rep-1}) then follows transporting $\PP$ on $\mu(\Sigma)$ by means of such homeomorphism and extending $\PP$ to the whole $B(\Omega)^+$ by letting $\PP = \nil$ on $B(\Omega)^+\setminus \mu(\Sigma)$. Eq.~(\ref{eq:rep-2}) follows by Proposition \ref{caratt. di E sui Borel}. By Theorem \ref{th:first}, $\PP^\prime = \Eo^\prime$. By Corollary \ref{Prop. misurab. di P} and Proposition \ref{Prop. supp. di P}, $B(\Omega)^+ \setminus P(\Omega)$ is a $\PP$-null Borel set.

Suppose $\PP_1$, $\PP_2$ are two PVMs on $B(\Omega)^+$ satisfying eq.~(\ref{eq:rep-1}). For $i=1,2$, let $\Sigma_i = \supp \PP_i$, and let $j_i : \Sigma_i \frecc B(\Omega)^+$ be the inclusion mapping. Therefore,
$$
\Eo (f) = \int_{\Sigma_i} [j_i (\mu) ] (f) \de \PP_i (\mu) \quad \forall f\in C_0 (\Omega) .
$$
By Theorem \ref{th:first}, there exists a homeomorphism $\Phi : \Sigma_1 \frecc \Sigma_2$ such that $j_1 = j_2 \circ \Phi$ and $\PP_2 = \PP_1 \circ \Phi^{-1}$. Hence, $\Sigma_1 = \Sigma_2$, $\Phi = \id$, and $\PP_1 = \PP_2$. This shows uniqueness of $\PP$ in eq.~(\ref{eq:rep-1}).

Suppose $\PP_1$, $\PP_2$ are two PVMs on $B(\Omega)^+$ satisfying eq.~(\ref{eq:rep-2}). If $f\in C_0 (\Omega)$, choose a sequence $\{ \phi_n \}_{n\in \N} \subset \linf{\Omega}$ as in the first half of the proof of Proposition \ref{misMark}. For all $\mu\in B(\Omega)^+$, we have $\mu(\phi_n) \to \mu(f)$, $|\mu(\phi_n)| \leq \no{\phi_n} \leq \no{f}$, and
$$
\int_{B(\Omega)^+} \mu(\phi_n) \de \PP_1 (\mu) = \int_{B(\Omega)^+} \mu(\phi_n) \de \PP_2 (\mu) .
$$
By dominated convergence theorem and uniqueness of the limit,
$$
\int_{B(\Omega)^+} \mu(f) \de \PP_1 (\mu) = \int_{B(\Omega)^+} \mu(f) \de \PP_2 (\mu) .
$$
Since this holds for all $f\in C_0 (\Omega)$, $\PP_1 = \PP_2$ follows.
\end{proof}

Theorem \ref{th:second} can be thought of as expressing the commutative observable 
$\Eo$ as an {\em operator average} over probability measures. There is still another 
representation of a commutative observable, which we now discuss, and which in a sense 
expresses $\Eo$ as a probability average over sharp observables (PVMs) and is thus 
complementary to the preceding  representation.  A detailed 
proof of this result can be found in \cite{Ali82,Ali84}. Here we only describe the 
representation.

Suppose $\Eo$ is an observable (not necessarily commutative), and consider the set $\Eo^\prime$ defined in eq.~(\ref{def-E'}). The double commutant $(\Eo^\prime)^\prime$, which we denote by ${\mathfrak A}  (\Eo)$, is  then a von Neumann algebra, and hence a weakly closed algebra of operators in $\lh$. (Recall that a von Neumann algebra $\mathfrak A$ is a
$\ast$-invariant set of bounded operators on $\hi$ which is an algebra under the operator
product and equal to its double commutant, ${\mathfrak A} = {\mathfrak A}^{\prime \prime}$. Such an
algebra is necessarily weakly closed). If $C^\ast (\Eo)$ is the $C^\ast$ algebra generated in $\lh$ by the set $\{ \Eo (f) \mid f\in C_0 (\Omega) \}$ and $I$, then clearly $C^\ast (\Eo) \subset {\mathfrak A}  (\Eo)$.

Next, consider the set $\widehat{\mathfrak S}$ of all sharp observables on $\Omega$ the ranges of which lie in ${\mathfrak A}  (\Eo)$. Thus, a sharp observable $\PP$ is an element of $\widehat{\mathfrak S}$ if and only if
$\PP(X )\in {\mathfrak A}  (\Eo)$ for all $X \in \bor{\Omega}$. There is a natural structure of a measure space on $\widehat{\mathfrak S}$, so that one can define ordinary Borel measures on it. Let $\overline{\nu}$ be such a positive measure on
$\widehat{\mathfrak S}$, with $\overline{\nu}(\widehat{\mathfrak S}  ) = 1$ (i.e., $\overline{\nu}$ is a probability
measure on $\widehat{\mathfrak S}  $). Then the operators $\overline{\Eo}(X ), X
\in \bor{\Omega}$, defined through the integral representation
\begin{equation*}
\overline{\Eo}(X) = \int_{\widehat{\mathfrak S}  }\PP(X)\de\overline{\nu}(\PP),
\quad \mbox{\rm or, symbolically,} \quad
\overline{\Eo} = \int_{\widehat{\mathfrak S}}\PP\de\overline{\nu}(\PP)
\end{equation*}
generate an observable. The convergence of the
above integral is in the weak topology of $\lh$, meaning that for arbitrary $\phi,\psi\in\h$,
\begin{equation*}
   \langle\phi\vert\overline{\Eo}(X )\psi\rangle =
   \int_{\widehat{\mathfrak S}  }\langle\phi\vert \PP(X )\psi\rangle\;
			     \de\overline{\nu}(\PP) \, .
\end{equation*}
The interesting fact about the above representation of the observable $\overline{\Eo}$ as a probability average over the sharp observables $\PP$, is that $\overline{\Eo}$ and $\overline{\nu}$ determine each other uniquely.

If $\Eo$ is commutative, there exists a unique probability measure $\overline{\nu}=\nu$ for which
$\overline{\Eo}= \Eo$. This result is stated precisely below.

\begin{theorem}\label{th:third}
Let $\Eo$ be a commutative observable on $\Omega$. There exists a unique probability Borel measure $\nu$ on $\widehat{\mathfrak S}$ such that
\begin{equation}\label{commPOVmeas3}
\Eo(X) = \int_{\widehat{\mathfrak S}  } \PP(X )\de \nu (\PP) \, .
\end{equation}
the integral converging weakly. The measure $\nu$ is unaltered if  ${\mathfrak A}  (\Eo)$ is replaced by any other von Neumann algebra which contains ${\mathfrak A}  (\Eo)$.
\end{theorem}

\subsection*{\bf A simple example}

The representation theorems for commutative observables can be illustrated by
means of a simple example. Let $\hi = L^2(\R, dx )$ and let $\nu$ be a fixed probability
measure defined on $\bor\R$. On $\hi$ we define the commutative observable
$\Eo (X), \; X \in \bor\R$, as the operator of multiplication by the function $x \mapsto
1_X * \nu (x)$, i.e., for any $\phi \in \hi$
\begin{equation*}
(\Eo (X)\phi )(x) = 1_X * \nu (x) \phi (x) \, , \quad \text{where} \quad
    1_X * \nu (x) = \int_\R 1_X (x - y )\; d\nu (y) \, ,
\end{equation*}
$1_X$ being the characteristic function of the set $X$. Defining the sharp observable $\PP_x$ for each $x \in \R$ and writing $\PP_0  = \PP$:
\begin{equation*}
(\PP_x (X) \phi)(y) = 1_X (y - x )\phi (y)\, , \qquad X \in\bor\R\, ,
\phi \in \hi,
\end{equation*}
we easily see that $\Eo (X)$ may be written in the form
\begin{equation}
\Eo (X) = \int_\R \PP_x (X) \;d\nu (x)\; ,
\label{first_rep_ex}
\end{equation}
as a weak integral. This is the representation given in eq.~(\ref{commPOVmeas3}).
On the other hand, it is also clear that $\Eo (X)$ may alternatively be
written in the form:
\begin{equation*}
(\Eo (X)\phi )(x) = \overline{\nu}_x (X) \phi (x)\; ,
\end{equation*}
again as a weak integral and
where $\overline{\nu}$ is the probability measure $\overline{\nu}_x  (X) =
\nu (x - X)$ and $\overline{\nu}_0 = \overline{\nu}$. Thus,
considering $x\mapsto \overline{\nu}_x (X)$ as a multiplication operator, for each
fixed $X \in \bor\R$, we get
\begin{equation} \label{second_rep_ex}
\Eo (X) = \int_\R \overline{\nu}_x (X) \de \PP (x) = \int_{B(\R)^+} \mu (X) \de \tilde{\PP} (\mu) \, ,
 \end{equation}
where $\tilde{\PP} = \PP \circ \Phi^{-1}$ with $\Phi : \R \frecc B(\R)^+$, $\Phi (x) = \overline{\nu}_x$. The last formula is the expression of $\Eo$ in the form (\ref{eq:rep-2}).
In (\ref{first_rep_ex}) the commutative observable $\Eo$ is written as an integral
over the sharp observables $\PP_x$, with respect to the probability measure $\nu$ carried by these sharp observables, while in (\ref{second_rep_ex}), $\Eo$ is written as an integral over the probability measures $\nu_x$, with respect to the sharp observable $\PP$, supported by the probability measures. If $f \in C_0 (\R)$ is a real function, we may think of it as
a sharp classical observable on the value space $\R$ while $f*\nu$ is an averaged out
version of $f$. Thus, since
\begin{equation*}
\Eo (f) = \int_\R \PP_x (f) \;d\nu (x) = \int_\R \overline{\nu}_x (f)\; d\PP (x)\; ,
\end{equation*}
we may call the first representation an unsharp quantum operator corresponding to a
sharp classical observable and the second representation as the sharp quantum operator of an unsharp classical observable.

\section{Fuzzy observables}\label{sec:fuzzy}

Let $\Omega_1$ and $\Omega_2$ be two lcsc spaces. Recall that a mapping $\mu:\Omega_2\to P(\Omega_1)$ is a \emph{Markov kernel} if $x\mapsto [\mu(x)](X)$ is measurable for every $X\in\bor{\Omega_1}$. By Proposition \ref{misMark}, this is equivalent to weak*-measurability and measurability of $\mu$. We will sometimes use the notation $\mu_{x}(X)\equiv [\mu(x)](X)$.

\begin{definition}\label{def:fuzzy}
Let $\Eo$ and $\Fo$ be two observables on $\Omega_1$ and $\Omega_2$, respectively. We say that $\Eo$ is a \emph{fuzzy observable with respect to} $\Fo$, or {\em fuzzy version of} $\Fo$, if there exists a Markov kernel $\mu : \Omega_2 \frecc P(\Omega_1)$ such that
\begin{equation}\label{eq:fuzzy-1}
p^{\Eo}_T (X)  = \int_{\Omega_2} [\mu (x)] (X) \de p^{\Fo}_T (x) \quad \forall X\in\bor{\Omega_1},\, T\in\sh \, .
\end{equation}
\end{definition}

The condition \eqref{eq:fuzzy-1} can be written as
\begin{equation}\label{eq:fuzzy-2}
\Eo(X)  = \int_{\Omega_2} [\mu (x)] (X) \de\Fo(x) \quad \forall X\in\bor{\Omega_1}
\end{equation}
or, equivalently,
\begin{equation*}
\Eo(f)  = \int_{\Omega_2} [\mu (x)] (f) \de\Fo(x) \quad \forall f\in C_0 (\Omega_1)
\end{equation*}
where, as usual, the integrals are understood in the weak sense.

Definition \ref{def:fuzzy} fits well with the concept of a fuzzy set. Namely, the mapping $[\mu (\cdot)] (X)$ is a fuzzy set for each $X\in\bor{\Omega_1}$. We can see the observable $\Eo$ as a composite mapping
\begin{equation*}
X \mapsto [\mu (\cdot)] (X) \mapsto \int_{\Omega_2} [\mu (x)] (X) \de\Fo(x) = \Fo ([\mu (\cdot)] (X)) \, .
\end{equation*}
The first mapping 'fuzzifies' outcome sets and then the observable $\Fo$ is applied on these fuzzy sets.

It is clear from \eqref{eq:fuzzy-2} that if $\Fo$ is a commutative observable, then also its fuzzy version $\Eo$ is commutative. In particular, all fuzzy versions of sharp observables are commutative. The results of Section \ref{sec:commutative} imply that also the converse is true, namely, all commutative observables are fuzzy versions of sharp observables.

Obviously, the representations for a commutative observable given in Section \ref{sec:commutative} are very close to the definition of a fuzzy version. If the set of outcomes $\Omega$ is compact, then by taking into account Proposition \ref{prop:compact} we see that \eqref{eq:borel} in Proposition \ref{caratt. di E sui Borel} gives $\Eo$ as a fuzzy version of $\PP$.
To see that the general (i.e. $\Omega$ not necessarily compact) representation form \eqref{eq:rep-2} formally satisfies the properties of Definition \ref{def:fuzzy}, let $\Eo$ be a commutative observable on $\Omega$. We then have the sharp observable $\PP$ given in Theorem \ref{th:second}. To show that $\Eo$ is a fuzzy version of $\PP$, we need to define a corresponding Markov kernel $\mu$ from $B(\Omega)^+$ to $P(\Omega)$. Fix a probability measure $\nu_0\in P(\Omega)$ and define $\mu$ as
\begin{equation*}
\mu:\quad \nu \mapsto \nu \quad \textrm{if $\nu\in P(\Omega)$} \, ; \quad \nu \mapsto \nu_{0} \quad  \textrm{if $\nu\in B(\Omega)^+\setminus P(\Omega)$}\, .
\end{equation*}
Then $\mu$ is measurable and hence a Markov kernel. As stated in Theorem \ref{th:second}, the set $B(\Omega)^+ \setminus P(\Omega)$ is a $\PP$-null set. Hence, we can write \eqref{eq:rep-2} in the form
\begin{equation*}
\Eo(X) = \int_{B(\Omega)^+} \nu(X) \de \PP(\nu) = \int_{B(\Omega)^+} [\mu(\nu)](X) \de \PP(\nu) \, .
\end{equation*}
Therefore, $\Eo$ is a fuzzy version of $\PP$.

\begin{example}\label{ex:covariant-markov}
Let us consider again the example discussed in the end of Section \ref{sec:commutative}. Hence, let $\nu$ be a fixed probability measure on $\R$. It defines a Markov kernel
\begin{equation*}
\R \ni x \mapsto \overline{\nu}_x \in P(\R)\, .
\end{equation*}
Here $\overline{\nu}_x$ is the probability measure $\overline{\nu}_x(X)=\nu(x-X)$. The fuzzification of a set $X\in\bor{\R}$ is now of the convolution form
\begin{equation*}
X \mapsto [1_X \ast \nu](\cdot) \, .
\end{equation*}

This type of Markov kernels can be defined in a similar way when $\Omega$ is some other lcsc group than $\R$. They typically arise in the situations where the observables under investigation have a covariance property. We refer to \cite{HeLaYl04} for further details.
\end{example}

Being a fuzzy observable with respect to some other observable is a relative notion - one has to specify two observables that we are comparing together. We can also say that $\Eo$ is a \emph{fuzzy observable} without further specification if it is a fuzzy version of some other observable $\Fo$ and this observable $\Fo$ is not a fuzzy version of $\Eo$.

There are then also those observables which are not fuzzy observables. The following result has been proven by Jen\v{c}ov\'{a} and Pulmannov\'{a} in \cite{JePu07}.

\begin{proposition}\label{prop:not-fuzzy-1}
A sharp observable $\PP$ is not a fuzzy observable if and only if the von Neumann algebra ${\mathfrak A}(\PP)$ is a maximal commutative von Neumann subalgebra of $\lh$.
\end{proposition}

This shows, in particular, that the usual sharp observables describing physical quantities such as position, momentum and spin of a particle and photon number of an electromagnetic field are not fuzzy observables. We will give an other formulation of this property in Section \ref{sec:joint}.

It is interesting to note that not all physical quantities are described by sharp observables and hence Proposition \ref{prop:not-fuzzy-1} cannot be applied to those cases. For instance,  it would be interesting to know whether some of the phase observables are fuzzy or not. (See Chapter III in \cite{OQP97} for an explanation of the phase observables and some other physical quantities not described by sharp observables.)

Markov kernels can add different kind of fuzziness to outcome sets. First of all, a Markov kernel can make outcome sets blurred. This happens, for instance, if in the situation of Example \ref{ex:covariant-markov} the probability measure $\nu$ is of the form $d\nu(x)=f(x)dx$ and $f$ is some smooth function (e.g. Gaussian function). This can typically model noise or imprecision in a measurement.

There is also another class of Markov kernels, which describe different kind of reduction. Namely, it can happen that a Markov kernel maps sets into sets (and not into genuine fuzzy sets), in which case there is thus no blurring. The possible reduction is now a consequence of the fact that two or more outcomes may become identified as one.

This latter kind of situation corresponds to a Markov kernel $\mu$ determined by a measurable function $\Phi:\Omega_2\to\Omega_1$, i.e., $\mu(x) = \delta_{\Phi(x)}$ ($\delta_x$ is the Dirac measure at $x$). The function $\Phi$ simply relabels the outcomes, possibly giving same label to several different outcomes. Eq.~\eqref{eq:fuzzy-2} now reads
\begin{equation}\label{eq:fuzzy-function}
\Eo(X)=\Fo(\Phi^{-1}(X)) \qquad \forall X\in\bor{\Omega_1} \, ,
\end{equation}
which we also write as
\begin{equation*}
\Eo=\Fo \circ \Phi^{-1} \, .
\end{equation*}

Our next proposition shows that fuzzification between two sharp observables can only be of the non-blurring form \eqref{eq:fuzzy-function}. This result follows immediately by combining Propositions 1 and 2 in \cite{Ho85} even in the more general case in which $\PP_2$ is an observable (not necessarily sharp). In this general form it has also been proved in \cite{JePu07}. In our case, it is a simple application of Theorem \ref{th:second}. We present this new proof in order to illustrate the role of the uniqueness part of Theorem \ref{th:second}

\begin{proposition}
Let $\PP_1$ and $\PP_2$ be sharp observables on $\Omega_1$, $\Omega_2$, respectively, with values in $\hi$. If $\PP_1$ is a fuzzy version of $\PP_2$, then there exists a measurable function $\Phi : \Omega_2 \frecc \Omega_1$ such that $\PP_1 = \PP_2 \circ \Phi^{-1}$.
\end{proposition}

\begin{proof}
Let $B_i^+ = B(\Omega_i)^+$. The map $\delta : \Omega_1 \frecc B_1^+$ given by
\begin{equation*}
[\delta (x)] (f) = f(x) \quad \forall x\in \Omega_1, f\in C_0 (\Omega_1)
\end{equation*}
is a homeomorphism of $\Omega_1$ into $\delta(\Omega_1)$, and $\delta(\Omega_1)$ is a locally closed subset of $B_1^+$.

Let $\mu : \Omega_2 \frecc B_1^+$ be a Markov kernel such that
\begin{equation}\label{ref}
\PP_1 (f) = \int_{\Omega_2} [\mu (x)] (f) \de \PP_2 (x) \quad \forall f\in C_0 (\Omega_1) \, .
\end{equation}

We claim that $\mu(x) \in \delta(\Omega_1)$ for $\PP_2$-almost all $x$. By eq.~\eqref{ref},
\begin{equation*}
\PP_1 (f) = \int_{B^+_1} \nu (f) \de (\PP_2 \circ \mu^{-1}) (\nu) \, .
\end{equation*}
On the other hand
\begin{equation*}
\PP_1 (f) = \int_{\Omega_1} [\delta (x)] (f) \de \PP_1 (x) = \int_{B^+_1} \nu (f) \de (\PP_1 \circ \delta^{-1}) (\nu) \, .
\end{equation*}
By the uniqueness part in Theorem \ref{th:second}, we have
\begin{equation*}
\PP_2 \circ \mu^{-1} = \PP_1 \circ \delta^{-1} \, .
\end{equation*}
It follows that $\mu^{-1} (B^+_1 \setminus \delta(\Omega_1) )$ is a $\PP_2$-null set, i.e., there exists $Z\in\bor{\Omega_2}$ such that $\PP_2 (Z) = \nil$ and $\mu(x) \in \delta(\Omega_1)$ for all $x\in \Omega_2 \setminus Z$.

Setting $\Phi(x) = \delta^{-1} (\mu(x))$ for all $x\in \Omega_2 \setminus Z$ and $\Phi(x) = {\rm constant}$ for all $x\in Z$, we obtain a measurable function $\Phi : \Omega_2 \frecc \Omega_1$, and
\begin{eqnarray*}
(\PP_2 \circ \Phi^{-1} )(X) & = & \PP_2 (\Phi^{-1} (X) \setminus Z) = \PP_2 (\mu^{-1} (\delta(X)) ) \\
& = & \PP_1 (\delta^{-1} (\delta(X)) ) = \PP_1 (X) \, .
\end{eqnarray*}
\end{proof}

\section{Joint measurability}\label{sec:joint}

In quantum measurement theory, the traditional compatibility relation between two observables is \emph{joint measurability}. If a pair of observables is jointly measurable, then their measurement outcome statistics can be recovered from the measurement statistics of a one single observable. The basic fact in quantum mechanics states that there are pairs of observables which are not jointly measurable. This opens up the room for many interesting questions.

The joint measurability relation and some related concepts are reviewed \cite{Lahti03}.  Here we just recall the definition and recommend this reference for more explanation. As in earlier sections, we always suppose that $\Omega_1 , \Omega_2$ are lcsc spaces.

\begin{definition}\label{def:joint}
Observables $\Eo_1$ and $\Eo_2$ are \emph{jointly measurable} if there exists an observable $\Go$, defined on the outcome space $\Omega_1\times\Omega_2$, such that
\begin{equation*}
\Go(X\times\Omega_2) = \Eo_1(X) \, , \qquad \Go(\Omega_1\times Y) = \Eo_2(Y)
\end{equation*}
for all $X\in\bor{\Omega_1}$ and $Y\in\bor{\Omega_2}$. In this case $\Go$ is a \emph{joint observable} of $\Eo_1$ and $\Eo_2$.
\end{definition}

Having a one moment reflection on Definition \ref{def:joint}, we notice that both $\Eo_1$ and $\Eo_2$ are special kind of fuzzy versions of $\Go$. Namely, if $\Phi_i : \Omega_1 \times \Omega_2 \frecc \Omega_i$ is the projection on the $i$-th factor, then $\Eo_i = \Go \circ \Phi_i^{-1}$.

This observation leads to a natural generalization of joint measurability. Namely, instead of requiring that the observables $\Eo_1$ and $\Eo_2$ are marginals of some observable, we can be interested whether they are fuzzy versions of one single observable. If this would be the case and $\Go$ is measured, we can calculate the measurement outcome distributions of $\Eo_1$ and $\Eo_2$. From this point of view, there is little difference whether  $\Eo_1$ and $\Eo_2$ are marginals or just fuzzy versions of $\Go$.

It thus seems that the joint measurability requirement is too restrictive and one should instead concentrate on the compatibility relation related to fuzzy versions.  However, this latter apparently more general relation is actually equivalent to joint measurability. This is the content of the following statement. The same result has been obtained in Proposition 3 of \cite{Ho85} using a different method.

\begin{proposition}\label{prop:joint-iff-fuzzy}
Two observables $\Eo_1$ and $\Eo_2$ are jointly measurable if and only if there exists an observable $\Go$ such that both $\Eo_1$ and $\Eo_2$ are fuzzy versions of $\Go$.
\end{proposition}

As we have discussed above, the 'only if' part of Proposition \ref{prop:joint-iff-fuzzy} is trivial since a marginal observable of $\Go$ is its particular type of fuzzy version. Before giving the proof of Proposition \ref{prop:joint-iff-fuzzy} in the general situation, we demonstrate the equivalence of these two conditions in the case of finite outcome sets.

Let $\Eo_1$ and $\Eo_2$ be two observables with finite outcome sets $\Omega_1=\{a_1,\ldots,a_m\}$ and $\Omega_2=\{b_1,\ldots,b_n\}$, respectively. To simplify the notation, we denote $\Eo_1(\{a_i\})\equiv \Eo_1(a_i)$, and so on. Suppose that $\Eo_1$ and $\Eo_2$ are both fuzzy versions of an observable $\Go$, which is defined on a finite outcome set $\Omega$. This means that we have Markov kernels $\mu$ and $\nu$ such that
\begin{equation*}
\Eo_1(a_i)=\sum_{x\in\Omega} \mu_{x}(a_i) \ \Go(x) \, , \qquad \Eo_2(b_j)=\sum_{x\in\Omega} \nu_{x}(b_j) \ \Go(x) \, .
\end{equation*}

For every $a_i\in\Omega_1,b_j\in\Omega_2,x\in\Omega$, denote
\begin{equation*}
\lambda_{x}(a_i,b_j)=\mu_{x}(a_i) \nu_{x}(b_j) \, .
\end{equation*}
It is clear that $0 \leq \lambda_{x}(a_i,b_j)\leq 1$ and
\begin{equation*}
\sum_{a_i,b_j} \lambda_{x}(a_i,b_j) = 1 \quad \forall x\in\Omega \, .
\end{equation*}
Thus, $\lambda$ is a Markov kernel.

We now define an observable $\widetilde{\Go}$ with the outcome set $\Omega_1\times\Omega_2$ as
\begin{equation*}
\widetilde{\Go}(a_i,b_j)=\sum_{x\in\Omega} \lambda_{x}(a_i,b_j) \ \Go(x) \, .
\end{equation*}
By its definition, $\widetilde{\Go}$ is a fuzzy version of $\Go$. Moreover, a simple calculation gives
\begin{equation*}
\sum_{a_i\in\Omega_1} \widetilde{\Go}(a_i,b_j) = \sum_{x\in\Omega} \sum_{a_i\in\Omega_1} \lambda_{x}(a_i,b_j) \Go(x) = \sum_{x\in\Omega} \nu_{x}(b_j) \Go(x)=\Eo_2(b_j)
\end{equation*}
and similarly
\begin{equation*}
\sum_{b_j\in\Omega_2} \widetilde{\Go}(a_i,b_j) =\Eo_1(a_i) \, .
\end{equation*}
We conclude that $\widetilde{\Go}$ is a joint observable for $\Eo_1$ and $\Eo_2$.

Before proving Proposition \ref{prop:joint-iff-fuzzy}, we need a small technical lemma.

\begin{lemma}\label{lemma:markov}
Let $\Sigma , \Omega_1 , \Omega_2$ be lcsc spaces. Suppose $\mu : \Sigma \frecc P(\Omega_1)$, $\nu : \Sigma \frecc P(\Omega_2)$ are Markov kernels. Then the mapping
\begin{equation*}
\lambda: \Sigma \frecc P(\Omega_1 \times \Omega_2) , \quad \lambda (h) = \mu (h) \otimes \nu (h)
\end{equation*}
is a Markov kernel.
\end{lemma}

\begin{proof}
If $f_1 \in C_0 (\Omega_1)$, $f_2 \in C_0 (\Omega_2)$, then the mapping
\begin{equation*}
h \mapsto [\lambda (h)] (f_1 \otimes f_2) = [\mu (h)] (f_1) [\nu (h)] (f_2)
\end{equation*}
is measurable by Proposition \ref{misMark}. If $g \in C_0 (\Omega_1\times\Omega_2)$, then there is a sequence $\{ g_n \}_{n\in\N} \subset C_0 (\Omega_1) \otimes C_0 (\Omega_2)$ such that $g_n \to g$ in the uniform norm. It follows that $[\lambda (h)] (g_n) \to [\lambda (h)] (g)$ for all $h$, and thus the mapping $h \mapsto [\lambda (h)] (g)$ is measurable. Another application of Proposition \ref{misMark} then gives the claim.
\end{proof}
\begin{proof}[Proof of Proposition \ref{prop:joint-iff-fuzzy}.]
The proof in the general case follows the same lines as the finite outcome example above.

Suppose that $\Omega_1,\Omega_2$ are the outcome spaces for $\Eo_1 , \Eo_2$, respectively. We assume that both $\Eo_1$ and $\Eo_2$ are fuzzy versions of the observable $\Go$, which is defined on the outcome space $\Omega$. We need to show that there exists a joint observable $\widetilde{\Go}$ of $\Eo_1$ and $\Eo_2$.

For every $X\in\bor{\Omega_1},Y\in\bor{\Omega_2},x\in\Omega$, denote
\begin{equation*}
\lambda (x) =  \mu (x) \otimes \nu (x) \, .
\end{equation*}
Then the mapping $x \mapsto \lambda (x)$ is a Markov kernel by Lemma \ref{lemma:markov}.

We define an observable $\widetilde{\Go}$ by formula
\begin{equation*}
\widetilde{\Go}(Z)=\int [\lambda (x)] (Z) \de \Go(x) \quad Z\in\bor{\Omega_1 \times \Omega_2} .
\end{equation*}
It follows directly from our construction that
\begin{equation*}
\widetilde{\Go}(X\times\Omega_2) = \Eo_1(X) \, , \qquad \widetilde{\Go}(\Omega_1\times Y) = \Eo_2(Y)
\end{equation*}
for all $X\in\bor{\Omega_1}$ and $Y\in\bor{\Omega_2}$. Hence, $\widetilde{G}$ is a joint observable of $\Eo_1$ and $\Eo_2$.
\end{proof}

In Section \ref{sec:fuzzy} we discussed about observables which are not fuzzy (in the absolute sense). This property has an equivalent formulation in terms of joint measurements, which is expressed in the next proposition.

\begin{proposition}
Let $\Eo$ be an observable. The following conditions are equivalent:
\begin{itemize}
\item[(a)] $\Eo$ is not a fuzzy observable.
\item[(b)] If $\Fo$ is another observable which is jointly measurable with $\Eo$, then $\Fo$ is a fuzzy version of $\Eo$.
\end{itemize}
\end{proposition}

\begin{proof}
Suppose that (a) holds. Let $\Fo$ be an observable which is jointly measurable with $\Eo$, and let $\Go$ be their joint observable. This means, in particular, that $\Eo$ is a fuzzy version of $\Go$. By the assumption, this can happen only if $\Go$ is a fuzzy version of $\Eo$. The relation of "being fuzzy version" is transitive (see e.g. \cite{Heinonen05}). Therefore, we conclude that $\Fo$ is a fuzzy version of $\Eo$. Thus, (b) holds.

Suppose then that (a) does not hold, i.e., $\Eo$ is a fuzzy observable. Then there is an observable $\Fo$ such that $\Eo$ is a fuzzy version of $\Fo$ but $\Fo$ is not a fuzzy version of $\Eo$. Since $\Fo$ is trivially fuzzy version of itself, we conclude from Proposition \ref{prop:joint-iff-fuzzy} that $\Eo$ and $\Fo$ are jointly measurable. Hence, (b) does not hold.
\end{proof}

\end{document}